\documentclass[11pt]{llncs}
\usepackage{makeidx}
\usepackage{graphicx}
 
\usepackage{ifthen}
\usepackage{microtype}%if unwanted, comment out or use option "draft"
\usepackage{boxedminipage}
\usepackage{amsmath}
\usepackage{amsfonts}
\usepackage{amssymb}
\usepackage{authblk}
\usepackage{breakcites}
\usepackage{chemscheme}
\usepackage{hyperref}

\usepackage{float}
\newfloat{Algorithm}{hbtp}{lop}[section]

\newcommand{\by}{\bar{y}}
\newcommand{\bx}{\bar{x}}

\newcommand{\pr}{\mathbf{Pr}}
\newcommand{\eps}{{\epsilon}}
\newcommand{\etal}{{\it et al.}}

\newcommand{\veps}{\varepsilon}

\graphicspath{{./Figures/}}

\begin{document}
 
\title{Sampling in Space Restricted Settings}
%
%\titlerunning{}  
% abbreviated title (for running head)
%                                     also used for the TOC unless
%                                     \toctitle is used
%
\author{Anup Bhattacharya \and Davis Issac \and Ragesh Jaiswal \and Amit Kumar}
%
%\authorrunning{Jaiswal et al.} 
% abbreviated author list (for running head)
%
%%%% list of authors for the TOC (use if author list has to be modified)
%\tocauthor{Ragesh Jaiswal}
%
\institute{Department of Computer Science and Engineering, \\
Indian Institute of Technology Delhi.\thanks{Email addresses: \email{\{csz128275, davis, rjaiswal, amitk\}@cse.iitd.ac.in}} 
}

\maketitle     % typeset the title of the contribution

\begin{abstract}
Space efficient algorithms play a central role in dealing with large amount of data. In such settings, one would like to analyse the large data using small amount of ``working space''. One of the key steps in many algorithms for analysing large data is to maintain a (or a small number) random sample from the data points. 
In this paper, we consider two space restricted settings -- (i) streaming model, where data arrives over time and one can use only a small amount of storage, and (ii) query model, where we can structure the data in low space and answer sampling queries.
In this paper, we prove the following results in above two settings:
\begin{itemize}
\item In the streaming setting, we would like to maintain a random sample from the elements seen so far. 
%All known strategies require huge amount of independent random bits, in some cases as many as the number of elements in the data.
We prove that one can maintain a random sample using $O(\log n)$ random bits and $O(\log n)$ space, where $n$ is the number of elements seen so far. We can extend this to the case when elements have weights as well. 
\item In the query model, there are $n$ elements with weights $w_1, \ldots, w_n$ (which are $w$-bit integers) and one would like to sample a random element with probability proportional to its weight. 
Bringmann and Larsen (STOC 2013) showed how to sample such an element using $nw +1 $ space (whereas, the information theoretic lower bound is $n w$). We consider the approximate sampling problem, where we are given an error parameter $\veps$, and the sampling probability of an element can be off by an $\veps$ factor. We give matching upper and lower bounds for this problem. 
\end{itemize}

\end{abstract}

\section{Introduction}

Space efficient algorithms are important when data is too large and cannot be stored in the working memory.
Such algorithms have become important with the increasing popularity of mobile devices.
These devices, in many cases have small amount of working memory.
Also, there is an increasing need to process the huge amount of data being generated over the internet for purposes of data mining.
In such scenarios, there is a need for analyzing the data in a {\em streaming} fashion.
This is popularly known as the streaming setting.
Note that in all these cases, the other resources such as the running time and amount of randomness\footnote{Note that typically a pseudorandom generator is used for generating random bits.} are equally important because these determine the power required for processing the data.
With the size of computing devices becoming smaller, power is becoming the most important resource to optimize in all such space-restricted settings.

\noindent
In this work, we look at the basic problem of random sampling.
The problem is very simple.
Given $n$ objects, the goal is to sample an object (or a few objects) uniformly at random.
This is called uniform sampling.
We can also consider non-uniform sampling where the objects come along with some weights and the goal is to sample an object with probability proportional to its weight.
We discuss these sampling problems in two space-restricted settings.
The first setting is the streaming setting, where the data items are available as a stream (i.e., the $i^{\textrm{th}}$ data item is available at time $i$) and one does not a priori know the number of data items that one should expect to see.
In such cases, maintaining a random sample at all times is more challenging than sampling in the classical setting where all data items are present in the memory.
This is partly because we cannot store all the data items in the stream due to space-restrictions.
The second space-restricted setting that we discuss is the query model.
Here, we talk about non-uniform sampling with respect to a given distribution.
In this model, one is allowed to pre-process the data and store a representation in small space so as to be able to quickly answer sampling queries.
Next, we discuss our results in the above two settings.

\subsection{Sampling in the streaming setting}
In this setting, the data items are available as a stream.
That is, the $i^{\textrm{th}}$ data item can be assumed to arrive at time $i$.
Here, we are interested in maintaining a uniformly random sample at all times.
We will generalise this for non-uniform sampling.
We would like the sampling algorithm to be one-pass and it should save only one data item (since each data item could be very large -- files/packets) in its working memory.

\noindent
The most basic method of doing this is called {\em reservoir sampling} and it proceeds in the following manner:
Let the items be denoted by $O_1, ...$ and the storage location used to store one item in the stream be denoted by $S$.
Store the first object $O_1$ in $S$.
Subsequently, for $O_i$, replace the previously stored item in $S$ with $O_i$ with probability $\frac{1}{i}$ and continue without changing $S$ with the remaining probability.
Whenever a sample is required, output the object stored in $S$.
Suppose $n$ objects have been seen until the time when the sample was produced.
The probability that $S$ stores $O_i$ is $\frac{1}{i} \cdot \frac{i}{i+1} \cdot ... \cdot \frac{n-1}{n} = \frac{1}{n}$.
So, we sample with the desired uniform distribution.
However, the amount of randomness required in this procedure is large.
Let us try to estimate the number of random bits required in this sampling procedure.
After $O_i$ arrives, the procedure will need random bits to decide whether the item stored in $S$ needs to be replaced with $O_i$.
This should happen with probability exactly $\frac{1}{i}$.
So, at least $\log{i}$ random bits will be required for this\footnote{Actually, more random bits might be required since $i$ might not be a power of $2$ and hence we might need to do {\em rejection sampling}. We can say that $O(\log{i})$ random bits are needed in expectation.}.
So, the number of random bits required for this procedure is at least $\sum_{i=1}^{n} \log{i} = \Omega(n \log{n})$ in expectation\footnote{We will discuss a more advanced sampling technique by Vitter that requires $\Omega((\log{n})^2)$ random bits in expectation.}.
This should be contrasted with the amount of random bits needed to uniformly sample in the classical setting where all the items are present in the memory.
%Suppose for the sake of simplicity of discussion that $n$ is a power of $2$.
In this case, in the classical setting where all the $n$ objects are present in the memory, we will need just $\log{n}$ random bits in expectation. 
%which is the minimum number of random bits required for sampling.
In this work, we address the gap in the amount of randomness required in the streaming versus the classical settings.
We will consider a model in which strict bounds on randomness may be defined as opposed to comparing expected amount of randomness required in classical and streaming settings.
We will first formalise the problem and then show that as far as amount of randomness is concerned, there is no gap in the streaming and the classical settings.
That is, even in the streaming setting, uniform sampling may be done using minimum number of random bits required in the classical setting.

\noindent
First, we note that upper bounds on randomness cannot be defined with respect to perfectly uniform sampling.
To see this, let us assume that $n>2$ is a prime number.
For the sake of contradiction, assume that a uniform sample can be generated using $r$ random bits for some finite $r$.
This means that there is a function $f : \{0, 1\}^r \to [n]$ such that for all $i, j \in [n]$, $|\{x | f(x) = i\}| = |\{x | f(x) = j\}|$.
This means that $2^r$ is divisible by $n$.
This is a contradiction since $n$ is a prime number.
One natural way of formalising the question of randomness efficiency with respect to uniform sampling is to allow the sampling algorithm to return a null answer (denoted by $\bot$) with certain small probability $\veps$.
This means that the sampling algorithm is allowed not to output any member of the set $\{1, ..., n\}$ with probability at most $\veps$.
Let us call this {\em uniform sampling with $\veps$-error}.
We can easily argue (see Section~\ref{sec:2}) that $\Omega(\log{\frac{n}{\veps}})$ random bits are required for uniform sampling with $\veps$-error.
Following is a simple algorithm that does uniform sampling with $\veps$-error using $O(\log{\frac{n}{\veps}})$-bits of randomness:
With the error parameter $\veps$ fixed, we first compute the smallest integer $r$ such that
$\lfloor 2^r/n\rfloor > 1 \quad$ and $2^r\ (mod\ n) \leq \veps \cdot 2^r$.
Let $k = \lfloor 2^r/n \rfloor$.
Consider a function $f$ that maps the first $k$ $r$-bit strings (ordered lexicographically) to $1$, the next $k$ strings to $2$ and so on.
The last $2^r\ (mod\ n)$ strings are mapped to $\bot$.
The sampling algorithm computes the function $f$ on the $r$ random bits and outputs the value of the function.

\vspace{0.1in}

\noindent
{\bf Our contributions}: 
Can the sampling ideas of the classical setting be extended to the streaming setting?
The answer is negative.
The main bottleneck in the streaming setting is that the value of $n$ is not known in advance
whereas the sampling algorithm in the streaming setting must maintain a sample at all times.
%One solution is reservoir sampling where fresh random bits are used after every new item arrives.
%However, as we have seen, this is costly in terms of the amount of randomness used.
In this work, we give a sampling algorithm that uses $O(\log{\frac{n}{\veps}})$-bits of randomness, uses $O(\log{\frac{n}{\veps}})$-space, and has a running time of $O(n + \log{\frac{n}{\veps}})$.
Moreover, the running time for processing each item is a constant except for the first item which is $O(\log{\frac{1}{\veps}})$.
\footnote{The running time is in terms of the number of arithmetic operations. If we take into account the number of bit-operations, then these bounds are larger by a multiplicative factor of $O\left((\log{\frac{n}{\veps}})^2\right)$.}
We also extend these results for non-uniform sampling.
Section~\ref{sec:2} gives details of these results.
It is important to point out that the lower bound on the number of random bits remains the same if the sampling algorithm is allowed to store more than one item from the stream.
Our sampling algorithm matches this lower bound while storing only one item from the stream.

\vspace{0.1in}

\noindent
{\bf Related work}: 
The initial techniques for sampling and reservoir sampling were discussed in \cite{knuth,vitter84,vitter85}.
Vitter's~\cite{vitter85} work was one of the early works on sampling in the streaming setting where the author was interested in sampling records that were stored in a magnetic tape by making a pass over the tape.
However, the computational resource that the author was interested in optimising was the running time of the sampling algorithm and not the amount of randomness or the space.
In fact, the author assumed that one can sample random numbers of arbitrary precision in the interval $[0, 1]$ in constant time.
Li~\cite{li} gave quantitative improvement over Vitter's work, again in terms of the running time bounds.
Park \etal~\cite{park} extended these ideas for sampling with replacement whereas Efraimidis and Spirakis~\cite{es06} did the same for weighted sampling.
Babcock \etal~\cite{bdm02} gave sampling algorithms where the sample is required to be among the most recent items seen in the stream.
They maintain a random sample over a moving window of the most recent items in the stream.

\vspace{0.1in}

\noindent
{\bf Comparison with Vitter's Reservoir Sampling}: 
Vitter's work on reservoir sampling~\cite{vitter85} is the most relevant previous work on this topic.
So, it is important to compare our results with those in~\cite{vitter85}.
We have already seen the most elementary reservoir sampling technique where the $i^{\textrm{th}}$ item is stored with probability $1/i$.
The expected number of random bits required for this is $O(\log{i})$ and so the expected number of random bits required for the overall algorithm is $O(n \log{n})$.
Note that this basic technique accesses fresh random bits for every item of the stream.
A somewhat  more sophisticated technique in~\cite{vitter85} reduces the number of times random bits are accessed by the sampling algorithm.
This technique works as follows: Suppose at time instance $i$, we have the $i^{th}$ item stored as the sample in the storage space $S$.
At this time, a positive integer $s$ is chosen from a particular probability distribution $f_i:\mathbb{Z} \to [0, 1]$.
This number $s$ denotes the number of stream items that the algorithm will skip before saving the item $(i+s+1)$.
This probability distribution is defined as $f_i(s) = \frac{i}{(i+s)(i+s+1)}$.
So randomness is required only for picking these ``skips".
It was shown in the paper that the expected number of times such skips need to be picked is $O(\log{n})$.
In order to sample from the distribution $f_i$, the paper assumes that one can uniformly sample a real number $u$ of arbitrary precision from $[0, 1]$.
One simple idea is to consider the cumulative distribution $F_i(s) = \sum_{i \leq s} f_i(s)$ and then pick the smallest value $s$ such that $F_i(s) \geq u$.

\noindent
Before further discussion regarding Vitter's work, let us draw a comparison between the models considered by our work and that in ~\cite{vitter85}.
First, in our model, randomness is consumed only in terms of random bits.
The reservoir sampling described above uses uniform random samples of arbitrary precision from $[0,1]$.
The second difference one should note is that both basic reservoir sampling and the one described above gives guarantees in terms of expected value of the randomness used. In our model, we are interested in the worst case number of random bits used given that the sampling algorithm is allowed to make some error.
So, in some sense, one may interpret our algorithm as a Monte Carlo algorithm and Vitter's reservoir sampling algorithm as a Las Vegas algorithm.

\noindent
In order to compare our results more closely, we need to remove the requirement of uniform samples from $[0, 1]$ in Vitter's algorithm.
So, the next question we address is whether one can sample from the distribution  $f_i$ using few random bits instead of uniform samples in $[0, 1]$.
Let us try to design an algorithm that sample $s$ from the distribution $f_i$ such that the expected number of random bits used by the algorithm is small.
Towards this, we first note that
$
\pr[s > i]  = 1 - \sum_{j \leq i}f_i(j)  = 1 - \sum_{j \leq i} \left( \frac{i}{i + j} - \frac{i}{i+j+1}\right) \leq 1/2.
$
We now consider the problem as sampling from the set $\{0, 1, ..., i+1\}$ as per a distribution $\mathcal{D}$, where $\forall j \leq i, \mathcal{D}(j) = f_i(j)$ and $\mathcal{D}(i+1) = \pr[s > i]$.
With respect to sampling $s$ from $f_i$, the $(i+1)^{th}$ item in the above problem corresponds to the case when $s > i$ and in this case we will draw a conditional sample from $\{i+1, i+2, ...\}$.
Note that if we can show that the expected number of random bits used in the above problem of sampling from $\{0, 1, ..., i+1\}$ is $R$, then the expected number of random bits required for sampling $s$ using $f_i$ will be $O(R)$.
So, let us just focus on the sampling problem above. For this we use the technique of Bringmann and Larsen~\cite{bl} (see section 2.1).
We will construct an array $A$ that contains numbers in $\{0, 1, ..., i+1\}$. 
$A$ contains the number $j$ exactly $\lfloor (i+2) \cdot \mathcal{D}(j)\rfloor + 1$ times. The sampling algorithm is as follows:
\begin{quote}
1. Pick a uniformly random $k \in \{1, ..., |A|\}$.\\
2. If $k = 1$ or $A[k] \neq A[k-1]$, then with probability $(1 - frac((i+2) \cdot \mathcal{D}(A[k])))$ go to step 1.\\
3. return $A[k]$.
\end{quote}

\noindent
Here $frac(x) = x - \lfloor x\rfloor$.
Bringmann and Larsen~\cite{bl} show that  the above sampling procedure returns a sample as per distribution $\mathcal{D}$ in constant 
expected time. 
Let us estimate the randomness required by this sampling procedure.
Note that $|A| \leq 2(i+2)$ and so step 1 costs $O(\log{i})$ random bits.
Also, since $\mathcal{D}(j) = \frac{i}{(i+j)(i+j+1)}$, the cost for simulating step 2 is $O(\log{i})$ random bits.
So, the expected number of random bits required in this sampling procedure is $O(\log{i})$.
As per our discussion earlier, this means that the expected number of bits required to sample $s$ from the distribution $f_i$ is $O(\log{i})$.
This further means that the expected number of bits required for Vitter's reservoir sampling algorithm is $O(\log^2{n})$.

\noindent
In the classical model where all the items are in the memory, the expected number of random bits required to sample is $O(\log{n})$.
So, within the model considered by Vitter's algorithm where one is interested in the expected number of random bits, there is a gap between the bounds in the classical and the streaming settings. An interesting question is whether this gap should exist.
Recall, that in our model where we are interested in the number of random bits when the sampling algorithm is allowed to err with small probability, we show there is no such gap between the classical and streaming setting.

\subsection{Succinct sampling}

The second space-restricted setting that we consider is a non-streaming setting where the set of elements are integers $\{1, ..., n\}$.
The most natural model of sampling in the non-streaming setting is the {\em query model}.
This is the model used by Bringmann and Larsen~\cite{bl} in their work.
Our work within this model may be interpreted as a natural extension of their work.
The inputs are $w$-bit integers $x_1, ..., x_n$.
The model includes a pre-processing step where appropriate data structures may be created.
Queries for producing a sample as per the weighted distribution are made and should be processed quickly using the data structures created in the pre-processing step.
The weighted distribution means that the query algorithm should output $i$ with probability $\frac{x_i}{\sum_j x_j}$.

\noindent
Bringmann and Larsen~\cite{bl} observed that the classical Walker's alias method~\cite{walker} in the word RAM model (here unit operations may be performed on words of size $w$ bits) has a pre-processing algorithm that runs in time $O(n)$, answers a sampling query in $O(1)$ expected time, and uses a storage of size $n (w + 2 \log{n} + o(1))$ bits.
In order to analyse the space usage, they defined a {\em systematic case} where the input is read-only and a {\em non-systematic case} where the input representation may be changed to reduce the total space.
The {\em redundancy} of a solution is the number of bits used in addition to the information-theoretic minimum required for storing the input.
Given this, the Walker's alias solution has a redundancy of $(2 n \log{n} + o(n))$-bits.
Bringmann and Larsen~\cite{bl} improved this and gave a solution in the systematic case where the preprocessing time is $O(n)$, expected query time is $O(1)$, and the redundancy is $n + O(w)$.
They also gave a solution that has $1$ bit of redundancy in the non-systematic case.
Furthermore, they showed optimality of their results.
However, all their results are for {\em exact} sampling.
In our work, we extend their work to {\em approximate} sampling in the word RAM model.

\noindent
In many realistic scenarios, we might not be required to sample exactly according to the weighted distribution $x_1, ..., x_n$.
One such scenario is the sampling based algorithms for $k$-means clustering such as the PTAS by Jaiswal \etal~\cite{jks12} where the algorithms are robust against small errors in sampling probability.
This indeed was the starting point of this work.
It may be sufficient to sample from a distribution such that the sampling probabilities are {\em close} to the exact sampling probabilities defined by the weights $x_1, ..., x_n$.
We will consider two models for closeness.
First is the {\em additive model} where the $i^{\textrm{th}}$ item's sampling probability may be between $\left(x_i/(\sum_j x_j) - \veps \right)$ and $\left( x_i/(\sum_j x_j) + \veps \right)$ for some small $\veps$.
Second is the {\em multiplicative model} where the $i^{\textrm{th}}$ item's sampling probability may be between $(1-\veps) \cdot \left(x_i/(\sum_j x_j) \right)$ and $(1+\veps) \cdot \left( x_i/(\sum_j x_j)\right)$ for some small $\veps$.
%We discuss these in the next few subsections.

\noindent
Before we state our results for approximate sampling, we should first understand the differences between exact and approximate sampling in terms of space usage.
Note that the information theoretic lower bound on the amount of space required to do exact weighted sampling given $n$ $w$-bit integers as input is $nw$.
\footnote{This is not a trivial observation since $x_1, ..., x_n$ and $x_1/2, ..., x_n/2$ give the same weighted distribution. See Lemma 5.1 in \cite{bl}.}\
However, in case of approximate sampling, the information theoretic bounds can be much lower since we can use some lossy representation of the inputs that does not effect the sampling probabilities too much but saves much space.
Given this, the non-systematic case (where data is not read-only and may be re-structured) seems more relevant than the systematic case (where the inputs are read-only and have to be retained).
So, in our work we discuss only the non-systematic case for approximate sampling.
Note that all the algorithms that we study have optimal pre-processing time of $O(n)$ and optimal query time of $O(1)$.

\vspace{0.1in}

\noindent
{\bf Our contributions}: We show that in the multiplicative model, the lower bound on the space requirement is $\Omega(n\log{w} + n \log{\frac{1}{\veps}})$.
We design a sampling algorithm and show that the space usage of our algorithm matches this lower bound.
In the additive model, we give similar results.
However, in this case our algorithms match the lower bound only when $\veps$ is a constant independent of $n$.

\vspace{0.1in}

\noindent
{\bf Related work}: 
Walker~\cite{walker} gave a solution for exact sampling in the classical setting.
Kronmal and Peterson~\cite{kp1979} improved the preprocessing time of Walker's method.
Bringmann and Panagiotou~\cite{bp2012} studied variants of sampling from discrete distribution problems.
All the above mentioned works used Real RAM model of computation.
Bringmann and Larsen \cite{bl} analysed Walker's alias method in Word RAM model of computation and also gave better bounds for exact sampling from discrete distribution problems.
Their work is most relevant to our current work on succinct sampling and our results may be regarded as a natural extension to \cite{bl}.

\newcommand{\prob}{{\mbox {\bf Pr}}}
\newcommand{\bits}{\{0,1\}}
\renewcommand{\sp}{{\hspace*{1 mm}}}
\newcommand{\append}{{\tt append}}

\section{Sampling in the Streaming Setting}\label{sec:2}

The input consists of a stream of distinct objects $O_1, O_2,...$, where the object $O_i$ can be thought of as arriving at time $i$. At any point of time, we would like to maintain a random sample from the set of objects seen so far. 
More formally, we would like to maintain a random variable $X_t$ for all time $t$ such that
$\prob[X_t = O_i]$ is the same for all $i=1, \ldots, t$. 
As mentioned in the introduction, this property cannot be achieved for all values of $t$. Therefore, the input also specifies a parameter $\veps$ -- the algorithm
is allowed to output a null object $\bot$ with probability at most $\veps$. 
Therefore, we want the following property to hold for all time $t$:
$$ \prob[X_t = \bot] \leq \veps, \prob[X_t=O_1] = \prob[X_t = O_2] = \cdots = \prob[X_t = O_t]. $$ 
We shall call such a sequence $X_t$ of random variables  {\em  uniform samples} (with error parameter $\veps$, which will be implicit in the discussion).

\noindent
In the setting of streaming algorithms, we would like to limit the space available to the algorithm. 
We allow the algorithm to store only one object at any point of time (besides some local variables) -- this is motivated by the fact that each object may be quite large (objects could be large files/packets etc.), and so it may not be feasible to store too many objects in the local memory of the program. 
%The generalisation to the case when the algorithm can store up to $k$ objects at any time follows in a straightforward manner.

\noindent
%We first discuss the standard technique, called {\em reservoir sampling}, of maintaining a uniform sample.
%As discussed in introduction, this method however needs $\Omega(n \log n)$ random bits (till time $n$). 
%This should be contrasted with 
Consider the amount of random bits needed to uniformly sample in the classical setting where all the $n$ items are present in the memory and we need a random sample among these items.
It is not difficult to show that $O \left(\log \frac{n}{\veps} \right)$ bits of randomness suffice (w.r.t. uniform sampling with $\veps$-error). 
In fact, it is also fairly easy to show that any algorithm (even in the
non-streaming setting) needs at least these many random bits. 
We give details of 
%reservoir sampling and 
the lower bound on number of random bits in Section~\ref{sec:back}. 
In Section~\ref{sec:result}, we show that 
%one can in fact do much better than reservoir sampling. 
we can maintain an exact sample with only $O \left(\log \frac{n}{\veps} \right)$ bits of randomness  (till time $n$). 
In Section~\ref{sec:wtd}, we extend this result to the weighted case.
%, and in Section~\ref{sec:k}, we argue how our results can be extended to the case when the streaming algorithm can store up to $k$ objects.

\subsection{Background}
\label{sec:back}

%We describe the idea of the reservoir sampling algorithm. 
%For simplicity,  assume that for any positive integer $i,$ we can generate a uniform random sample from the set $\{1, \ldots, i\}$ using $O(\log i)$ bits of randomness.
%At time $t$, suppose we have a uniform sample $X_t$. 
%When the object $O_{t+1}$ arrives at time $t+1$, we set $X_{t+1} = O_{t+1}$ with probability $\frac{1}{t+1}$, otherwise $X_{t+1}$ remains equal to $X_t$. 
%As shown in the introduction, the sequence $X_t$ form uniform samples. 
%Also the total number of bits of randomness needed till time $n$ is $\Omega(n \log n)$.
We consider the off-line problem of generating a uniform sample with error parameter $\veps$ from the set of objects $O_1, \ldots, O_n$. The proof of the next lemma is given in the Appendix.

\begin{lemma}
\label{lem:offline}
We can generate a uniform sample with error parameter $\veps$  from a set of $n$ distinct objects using $O(\log \frac{n}{\veps})$ random bits. 
Further, any algorithm for generating such a sample must use $\Omega(\log \frac{n}{\veps})$ random bits.
\end{lemma}

\noindent
The above idea for upper bound does not work in the streaming setting.
The main problem  in the streaming setting is that the value of $n$ is not known in advance -- the algorithm needs to maintain a uniform sample at {\em all} times.
One solution is reservoir sampling where fresh random bits are used after every new item arrives.
However, as we have seen, this is costly in terms of the amount of randomness used.
In the next section,  we discuss a sampling algorithm in the streaming setting that uses $O(\log{\frac{n}{\veps}})$ random bits till time $n$, and hence, matches the lower bound result mentioned above.

%%%%%%%%  New Section  %%%%%%%%
\subsection{Uniform samples in the streaming setting}\label{sec:result}

Let us try to understand some of the challenges of designing sampling algorithms in the streaming setting. Recall that $X_t$ is the random object
maintained by the algorithm at time $t$. Since the algorithm is allowed to store only one object at any time, it does not store any other object at
time $t$. 
At time $t+1$, when $O_{t+1}$ arrives, the algorithm has only three choices for $X_{t+1}$ -- $X_t, O_{t+1}$ or $\bot$. We shall use
$r_t$ to denote the number of random bits used by our algorithm till time $t$. Given a sequence $x_t$ of $r_t$ random bits, let $f_t(x_t)$ denote
the object stored by the algorithm at time $t$, i.e., $X_t = f_t(x_t)$. Note that the functions $f_t$ need to satisfy a  ``consistency'' property:
if $x \in \bits^{r_{t}}$ is a prefix of a string $y \in  \bits^{r_{t+1}},$ then $f_{t+1}(y)$ is either $f_t(x)$ or $O_{t+1}$ or $\bot$. 
We now describe our algorithm that we call the {\em doubling-chopping} algorithm.

\subsubsection{The algorithm}
For each time $t$ and  $i \in \{1, \ldots, t\} \cup \{ \bot \}$, the algorithm will maintain an ordered set $H^t_i \subseteq \bits^{r_t}$  of strings $x_t$ for which $f_t(x_t)= O_i$ (or $\bot$).  
Of course, this will lead to large space complexity -- we will later show that these sets can be maintained implicitly. 
Initially, at time $0$, $H^0_{\bot}=\emptyset$ and $r_0 = 0$. 
We first describe the doubling step in Figure~\ref{fig:double-chop}. 
The goal of this step is to ensure that $2^{r_t}$ stays larger than $\frac{(t+1)^2}{\veps}$. 
Whenever this does not happen, we increase the value of $r_t$ to ensure that this is the case.
The functions $f_t$ are updated accordingly -- they just look at the first $r_t$ bits of the input.

  \begin{figure}[ht]
    \begin{boxedminipage}{0.39\linewidth}
     {\bf Algorithm {\tt Double($t$)} :} \medskip\\
       \sp \sp 1. $r_t \leftarrow r_{t-1}$. \\ 
       \sp \sp 2. {\bf For} $i \in \{1, ..., t-1\}\cup \{\bot\}$ \\
       \sp \sp \sp \sp \sp \sp $\bullet$ Initialize $H_i^t \leftarrow H_i^{t-1}$. \\
       \sp \sp 3. {\bf While} $2^{r_t} < \frac{(t+1)^2}{\veps}$ \\
       \sp \sp \sp \sp \sp (i) $r_t \leftarrow r_t + 1.$ \\
       \sp \sp \sp \sp \sp  (ii) {\bf For} $i \in \{ 1, ..., t-1 \} \cup \{\bot\}$ \\
       \sp \sp \sp \sp \sp \sp \sp \sp \sp \sp $\bullet$ Initialize $H \leftarrow \emptyset$. \\
       \sp \sp \sp \sp \sp \sp \sp \sp   \sp \sp $\bullet$ {\bf For} each $x \in H^t_i$ in order \\
       \sp \sp \sp \sp \sp \sp \sp \sp \sp \sp \sp \sp \sp \sp \sp \sp \sp \sp append $0x$ to $H$. \\
       \sp \sp \sp \sp \sp \sp \sp \sp \sp \sp $\bullet$ {\bf For} each $x \in H^t_i$ in order \\
       \sp \sp \sp \sp \sp \sp \sp \sp \sp \sp \sp \sp \sp \sp \sp \sp \sp \sp append $1x$ to $H$. \\
       \sp \sp \sp \sp \sp \sp \sp \sp \sp \sp $\bullet$ $H_{i}^{t} \leftarrow H$. 
      \end{boxedminipage}      
       \quad
       \begin{boxedminipage}{0.57\linewidth}
     {\bf Algorithm {\tt Chop($t$)} :} \medskip\\
       \sp \sp 1. {\bf For} every $i \in \{1, \ldots, t-1\}$ \\
       \sp \sp \sp \sp \sp \sp \sp  Define $T^{t}_i \leftarrow $  last $\left( |H^t_i| - \lfloor 2^{r_t}/t \rfloor \right)$ strings in $H^t_i$. \\
       \sp \sp \sp \sp \sp \sp \sp Define $H^t_i \leftarrow H^t_i \setminus T^t_i.$ \\
       \sp \sp 2. Initialize $T \leftarrow \emptyset$. \\
       \sp \sp 3. {\bf For} $i =1, \ldots, t-1$ \\
       \sp \sp \sp \sp \sp \sp \sp $T \leftarrow \append(T, T^t_i)$. \\
       \sp \sp 4. {\bf If} $|T| > \lfloor 2^{r_t}/t \rfloor$ \\
       \sp \sp \sp \sp \sp \sp \sp \sp \sp (a) $T^t_{t} \leftarrow$ last $ \left(|T| - \lfloor 2^{r_t}/t\rfloor \right)$ strings in $T$. \\
       \sp \sp \sp \sp \sp \sp \sp \sp \sp (b) $H^{t}_{t} \leftarrow  T \setminus T^{t}_{t}$ and $H^{t}_{\bot} \leftarrow \append(H^t_{\bot},
       T^t_t)$ \\
       \sp \sp \sp \sp \sp \sp {\bf Else} \\
       \sp \sp \sp \sp \sp \sp \sp \sp \sp  (i) $T^t_{\bot} \leftarrow$  last $\lfloor 2^{r_t}/t\rfloor-|T|$  strings of $H^t_{\bot}$. \\
        \sp \sp \sp \sp \sp \sp \sp \sp \sp  (ii) Set $H^{t}_{\bot} \leftarrow  H^t_{\bot} \setminus T^t_{\bot}$ and  $H^{t}_{t} \leftarrow \append(T, T^t_{\bot})$.
      \end{boxedminipage}
      \caption{The doubling and chopping steps}
       \label{fig:double-chop}
    \end{figure}

\noindent
Note that after we call the algorithm {\tt Double}, the new $r_{t} - r_{t-1}$ bits do not participate in the choice of random sample. 
In Step 3 of the {\tt Double} algorithm, the set $H^{t}_i$ is an ordered list -- ``append'' adds an element to the end of the list.

\noindent
The next step, which we call the chopping step, shows how to modify the function $f_t$ so that
some probability mass moves towards $O_{t}$. 
The algorithm is described in Figure~\ref{fig:double-chop}. 
The function $\append(T_1, T_2)$ takes two ordered lists and outputs a new list obtained by first taking all the elements in $T_1$ followed by the elements in $T_2$ (in the same order). 
The algorithm maintains the sets $H_i^t$, where $i \in \{1, \ldots, t\} \cup \{\bot\}$. 
Given these sets, the function $f_t$ is immediate. 
If the string $x \in \bits^{r_t}$ lies in the set $H_i^t$, then $f_t(x) = i$.

\noindent
To summarise, at time $t>1$, we first call the function {\tt Double$(t)$} and then the function {\tt Chop$(t)$} (at time $t=1$ we only call {\tt Double$(1)$}).
It is also easy to check that the functions $f_t$ satisfy the consistency criteria.
\begin{lemma}
\label{lem:cons}
Suppose $x \in \bits^{r_{t-1}}$ and $y \in \bits^{r_t - r_{t-1}}$. Then, $f_t(yx)$ is either $f_{t-1}(x)$ or $O_t$ or $\bot$.
\end{lemma}

\begin{proof}
Let $x$ and $y$ be as above. 
Suppose $x \in H^{t-1}_i$ (and so, $f_{t-1}(x)=i$). 
After the {\tt Double($t$)} function call, $yx \in H^t_i$.
Now consider the function {\tt Chop($t$)}. 
After  Step 1, if $yx \notin H^t_i$, then it must be the case that $yx$ gets added to the set $T$.
Now, notice that the strings in $T$ get added to either $H^t_t$ or $H^t_\bot$. 
This proves the lemma.\qed
\end{proof}

\noindent
The lemma above implies  that  we can execute the algorithm by storing only one object at any time.  Now, we show that the number of random bits
used by the algorithm is small.
\begin{lemma}
\label{lem:random}
The total number of random bits used by the algorithm till time $n$ is $O( \log \frac{n}{\veps})$.
\end{lemma}

\begin{proof}
Till time $n$, the algorithm uses at most $r_n$ bits. Now, by definition, $2^{r_n} \leq \frac{2(n+1)^2}{\veps}$. This proves the lemma.\qed
\end{proof}

\noindent
Now, we prove the correctness of the algorithm.
\begin{lemma}
\label{lem:correct}
For all time $t$, and $i \in \{1, \ldots, t\},$ $|H^t_i| = \left \lfloor \frac{2^{r_t}}{t} \right \rfloor$, and $|H^t_\bot| \leq \veps \cdot 2^{r_t}$.
\end{lemma}
\begin{proof}
The proof is by induction on $t$. 
The base case ($t=0$) is true vacuously. 
Now suppose the lemma is true for $t-1$. 
At time $t$, we first call {\tt Double($t$)}. 
For each $x \in H^{t-1}_i$, we just append all bit strings of length $r_t - r_{t-1}$ to it and this set of strings to $H^t_i$. 
Therefore, when this procedure ends, $|H^t_i| = 2^{r_t-r_{t-1}} \cdot \lfloor \frac{2^{r_{t-1}}}{t-1} \rfloor$, for $i=1, \ldots, t-1$ (using induction hypothesis) and we have
\begin{eqnarray*}
|H_i^{t}| = 2^{r_t-r_{t-1}} \cdot \left \lfloor \frac{2^{r_{t-1}}}{t-1} \right \rfloor 
\geq 2^{r_t-r_{t-1}} \cdot \left( \frac{2^{r_{t-1}}}{t-1} - 1\right) 
\geq \frac{2^{r_t}}{t} \quad (\textrm{since $2^{r_{t-1}} \geq t^2/\veps$})
\end{eqnarray*}
In Step 1 of the procedure {\tt Chop($t$)}, we ensure that $|H^t_i|$ becomes $\lfloor \frac{2^{r_t}}{t} \rfloor$ (this step can be done, because the $|H^t_i|$ was at least $\lfloor \frac{2^{r_t}}{t} \rfloor$). 
After this step, we do not change $H^t_i$ for $i=1, \ldots, t-1$, and hence, the induction hypothesis is true for these sets. 
It remains to check the size of $H^t_t$ and $H^t_\bot$.

\noindent
First assume that $|T| \geq \lfloor \frac{2^{r_t}}{t} \rfloor$. 
In this case, $H^t_t$ gets exactly $\lfloor \frac{2^{r_t}}{t} \rfloor$ elements.
Now suppose  $|T| < \lfloor \frac{2^{r_t}}{t} \rfloor$. 
First observe that $H^t_\bot$ and $T$ are disjoint. 
Since all strings not in $H^t_i, i = 1, \ldots, t-1$ belong to either $H^t_\bot$ or $T$, it follows that
$$|H^t_\bot| + |T| = 2^{r_t}  - (t-1) \cdot \left \lfloor \frac{2^{r_t}}{t} \right \rfloor \geq \left \lfloor \frac{2^{r_t}}{t} \right \rfloor .$$
Therefore, $|H^t_\bot|$ is at least $\lfloor \frac{2^{r_t}}{t} \rfloor -|T|$, and Step 4(i) in this case can be executed. 
Clearly, $|H^t_t|$ becomes $\lfloor \frac{2^{r_t}}{t} \rfloor$ as well.
Finally,
$$|H^t_\bot| = 2^{r_t} - t \cdot \left \lfloor \frac{2^{r_t}}{t} \right \rfloor \leq 2^{r_t} - t \left( \frac{2^{r_t}}{t} - 1 \right) = t \leq \veps \cdot 2^{r_t},$$
where the last inequality follows from the definition of $r_t$.
\qed
\end{proof}

%%%%%%%%  New Section  %%%%%%%%
\subsubsection{Space Complexity}
Note that the use of the sets $H^t_i$ in our algorithm was just for sake of clarity. 
We need not maintain these sets explicitly.
For the current random string $x$ (at time $t$), we just need to keep track of the set $H^t_i$ to which it belongs -- call this set $L(x)$ (the {\em location} of $x$). 
In fact, not only we will keep track of $L(x)$, but we will also keep track of the {\em rank} of $x$ in the set $L(x)$ -- recall that the sets $H^t_i$ are ordered lists, and so, the rank of an element is its position in this order.
In addition, we will also keep track of $|H_i^t|$ for $i \in \{1, ..., t\}\cup \{\bot\}$. 
Note that this includes saving only two numbers since $|H_1^t| = ... = |H_t^t|$.
The pseudocodes of our algorithms for implementation purposes are given in the Appendix.

\begin{lemma}
\label{lem:space}
For  every time $t$, the location and the rank of the current random string $x_t$ can be maintained using $O(\log \frac{t}{\veps})$ space.
\end{lemma}

\begin{proof}
Suppose the statement is true for $t-1$, and say, $x_{t-1} \in H^{t-1}_i$. During {\tt Double($t$)}, we will append a random string $y \in
\bits^{r_t-r_{t-1}}$ to $x_{t-1}$, i.e., $x_t = y x_{t-1}$. 
For every string preceding $x_{t-1}$ in $H^{t-1}_i$, we will add $2^{r_t - r_{t-1}}$ strings to $H^t_i$. 
Hence, one can easily determine the rank of $x_t$ in $H^{t}_i$. Using this fact, we can check whether $x_t$ gets transferred to $T$ or not during Step 1 of {\tt Chop($t$)}. 
%Moreover, since the size of the sets $H^t_i$ (at the beginning of Step 1) is $2^{r_t - r_{t-1}} \cdot \lfloor \frac{2^{r_t}}{t-1} \rfloor$ (note that we have this stored), we can even calculate the rank of $t$ in $T$. 
Moreover, since we know the size of the sets $H^t_i$ (at the beginning of Step 1), we can even calculate the rank of $t$ in $T$. 
Since we also know the size of $H^t_\bot$, we can check if $x_t$ gets transferred to $H^t_t$ or $H^t_\bot$ in Step 4 (and its rank in this set).
The space needed by the algorithm is proportional to $r_t$, which is $O(\log \frac{t}{\veps})$.
\qed
\end{proof}

%%%%%%%%  New Section  %%%%%%%%
\subsubsection{Running time}
Finally, we analyse the running time of the algorithm after $n$ time steps. 
The total number of iterations of {\tt While} loop in 
{\tt Double($n$)} is at most $r_n$, i.e., $O(\log \frac{n}{\veps})$. The time taken by {\tt Chop($n$)} is constant number of arithmetic operations -- Step 1 is constant number of operations. 
If the string happens to be in $T$, its rank can be computed in constant number of operations,
and similarly for Step 4. 
Therefore, the total running time till time $n$ is $O(n + \log{\frac{n}{\veps}})$. 
However, the running time of {\tt Double} per unit time step can be more than a constant. 
It is not difficult to see that except for the first time step (when we need to make $r_1 = \lceil \log \frac{4}{\veps} \rceil$), $r_{t+1}$ is at most $r_t + 2$. 
Therefore, except for the first two time steps, the running time per time step is constant number of arithmetic operations.  
Figure~\ref{fig:sampling2} show the operations performed by the sampling algorithm after the arrival of the second item.

\noindent
Note that in the above analysis, the running time is in terms of the number of arithmetic operations. 
However, as $n$ grows, the number of bit operations is a more relevant measure.
Since at each time step, the arithmetic operations are over numbers of size $O(\log{\frac{n}{\veps}})$-bits, the total running time in terms of bit operations will be $O\left( \left(\log{\frac{n}{\veps}}\right)^2 \cdot (n + \log{\frac{n}{\veps}}) \right)$ and the per item running time will be $O\left( \left(\log{\frac{i}{\veps}}\right)^2\right)$ (w.r.t. the $i^{\textrm{th}}$ item).

\begin{figure}
\centering
\includegraphics[scale=0.5]{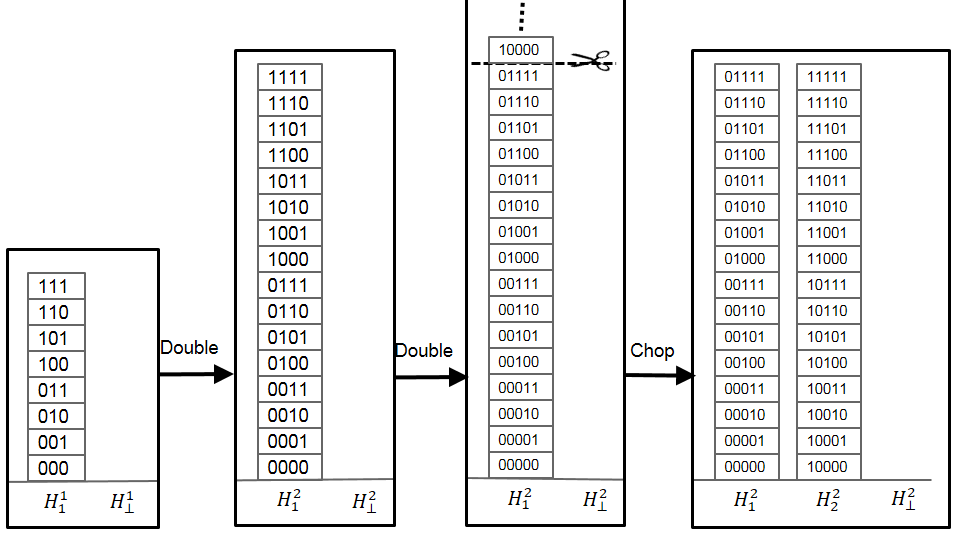}
\caption{The figure shows simulation of the doubling-chopping algorithm at time $t = 2$ when the value of $\veps = 1/2$.}
\label{fig:sampling2}
\end{figure}

%%%%%%%%  New Section  %%%%%%%%
\subsection{Weighted sampling}\label{sec:wtd}
So far we have discussed uniform sampling.
We can now talk about a more general setting where items in the stream come along with an associated integer weight $w_i$ and after seeing $n$ elements the sampling algorithm should be prepared to output the $i^{th}$ item with probability proportional to $w_i$.
Further, as in the uniform sampling case, the algorithm is allowed to output $\bot$ with probability at most $\veps$.
More specifically, let the algorithm output $\bot$ with probability $p$.
Then $p \leq \veps$ and the probability that it outputs item $i$ is given by $(1 - p) \cdot \frac{w_i}{\sum_i w_i}$.
The simplest solution is to consider  $w_i$ copies of item $i$ and simulate our sampling algorithm.
The number of random bits that are required is $\log{\left(\frac{\sum_i w_i}{\veps}\right)}$.
Given that each of the weights $w_1, ...$ is a $w$-bit integer, we get that the upper bound is $O(\log{\frac{n \cdot 2^w}{\veps}}) = O(w + \log{\frac{n}{\veps}})$.
From Lemma~\ref{lem:offline}, we know that the lower bound is $\Omega(\log{\frac{n}{\veps}})$ when $w_1 = w_2 = ... = 1$.
Furthermore, given that $w_1 = 1$ and $w_2 = 2^{w}-1$, any sampling algorithm would need at least $w$ bits for uniform sampling. This gives another lower bound of $\Omega(w)$. 
From the last two statements, we get that the lower bound on the number of random bits required for weighted sampling is $\Omega(w + \log{\frac{n}{\veps}})$ which matches with our upper bound.
In this setting, simple space/time optimisations lead to a sampling algorithm with running time $O(n + w + \log{\frac{1}{\veps}})$ (with per item time $O(w)$) and space $O(w + \log{\frac{n}{\veps}})$.

\section{Succinct (approximate) Sampling}

In this section, we consider approximate sampling in the succinct data-structure model. 
Recall that we are given a set of $n$ elements, labelled $1, \ldots, n$, and weights $x_1, \ldots, x_n$ associated with these $n$ elements respectively. 
Each weight $x_i$ is a $w$-bit integer. We will assume throughout the discussion that $w = o(n)$ which is a reasonable assumption since $w$ is typically a single precision ($w = 32$) or double precision ($w = 64$) number.
We are allowed to store a suitable representation of these
weights such that we can perform sampling efficiently. More formally, let $p_i$ denote $x_i/({\sum_j x_j})$. Given an
error parameter $\veps$, we consider two notions of approximate sampling -- multiplicative and additive. In the multiplicative
model, we are required to output a random element such that the probability that $i$ is output lies in the range $[p_i(1-\veps),
p_i(1+\veps)]$. In the additive model, the corresponding probability of $i$ lies in the range $[p_i - \veps, p_i + \veps]$.

\subsection{Approximate sampling: multiplicative model}
In this section, we give upper and lower bounds on the amount of space needed to perform approximate sampling with (multiplicative) error $\veps$.
We first discuss the upper bound by giving our sampling algorithm and then give matching lower bounds. 
For simplicity, we assume that $\veps$ is a power of $2$ (this only affects the bounds by a constant factor).

\vspace{0.1in}

\noindent
{\bf Upper Bound}: 
For each $i$, let $f_i$ denote the location of the most significant bit (MSB) in $x_i$ which is 1 (i.e., the first $f_i-1$ bits of $x_i$ are 0). 
Let $x_i'$ denote the number obtained by taking the first $f_i + \log{\frac{2}{\eps}}$ bits of $x_i$ followed by $\left(w - f_i - \log{\frac{2}{\veps}}\right)$ 0's. 
It is easy to check that exact sampling with respect to the weights $x_i'$ leads to approximate sampling with respect to $x_i$ with error at most $\veps$.

\begin{lemma}
\label{cl:mult}
For all $i$,
$(1-\veps) \cdot p_i \leq \frac{x_i'}{\sum_j x_j'} \leq (1+\veps) \cdot p_i$.
\end{lemma}

\begin{proof}
Observe that for all $i$, $x_i - x_i' \leq \frac{\veps}{2} \cdot x_i$, which implies that $x_i \geq x_i' \geq (1 - \veps/2) \cdot x_i$. Using this fact, we get
$ (1-\veps) p_i \leq \frac{x_i(1-\veps/2)}{\sum_j x_j} \leq \frac{x_i'}{\sum_j x_j}
\leq \frac{x_i'}{\sum_j x_j'} \leq \frac{x_i}{(1-\veps/2)\sum_j x_j} \leq (1+\eps) p_i$.
\qed
\end{proof}

\noindent
Therefore, it is enough to run an exact sampling algorithm with weights $x_i'$ for all $i$. For this, we use the algorithm of
Bringmann and Larsen~\cite{bl} with respect to $x_i'$.
The space needed by this algorithm is $O(n+w')$, where $w'$ is the number of bits needed to store any of the weights. In our case, $w'$ can be as high as $w$, and so the space needed by their algorithm is $O(n+w)$.
There is one catch though: we need to store all the $x_i'$ using the same number of bits, and using $w$ bits would be a waste of space. Instead
we store each $x_i'$ as a tuple -- we first store the value of $f_i$ and then the value of the next $\log \frac{2}{\eps}$ bits. Note that
this representation uses $(\log w + \log  \frac{2}{\eps})$ bits for each of the $x_i'$. It is not difficult to check that the algorithm of
Bringmann and Larsen~\cite{bl} works with this representation as well. Thus, the total space needed by our algorithm is
$O(n \log w + n \log  \frac{2}{\eps})$.

\vspace{0.1in}

\noindent
{\bf Lower bound}:
In this section, we give lower bound on the amount of space needed for approximately sampling the elements with error $\veps$.
We say that a distribution given by $(y_1, ..., y_n)$ is $\veps$-close to a distribution given by $(x_1, \ldots , x_n)$ if $\forall i, (1-\eps)\cdot \frac{x_i}{\sum_j x_j} \leq \frac{y_i}{\sum_j y_j} \leq (1+\eps)\cdot \frac{x_i}{\sum_j x_j}$.
In such a case, $(y_1, \ldots , y_n)$ may be used to represent the distribution $(x_1, ..., x_n)$.
To get a lower bound on the space, we will  estimate the size of a set of tuples $S \subseteq [\bits^w]^n$ such that for any tuple $\bar{x} \in [\bits^w]^n$  there exist at least one element $\bar{y}$ in $S$ such that $\bar{y}$ is $\veps$-close to $\bar{x}$.
Let $\mathcal{S}$ denote the minimum amount of space needed.
We get a lower bound on $\mathcal{S}$ using the next two lemmas.
The proofs of these lemmas are given in the appendix.

\begin{lemma}\label{lemma:1}
$\mathcal{S} \geq n \log{\frac{1}{\eps}} - w - \log{n} - n $.
\end{lemma}

\begin{lemma}\label{lemma:2}
$\mathcal{S} \geq n \log{w} - n \log{4 (1 + 2 \eps)} - \frac{w}{2} \log{(e^2 n)}$.
\end{lemma}

\noindent
{\bf Comparing upper bounds with lower bounds} The upper bound that we obtained on the space requirement was $(n \log{\frac{2}{\eps}} + n \log{w})$.
We break the comparison into the following two parts:
\begin{enumerate}
\item $\frac{1}{\eps} > w$: In this case, the upper bound is $O(n \log{\frac{1}{\eps}})$. Using Lemma~\ref{lemma:1}, we get that the lower bound is $\Omega(\log{\frac{1}{\eps}})$ assuming $w = o(n)$ and $\eps \leq 1/2$.

\item $\frac{1}{\eps} \leq w$: In this case, the upper bound is $O(n \log{w})$. Using Lemma~\ref{lemma:2}, we get that the lower bound is $\Omega(n \log{w})$ assuming $w = o(n)$ and $\eps \leq 1/8$.
\end{enumerate}
So, we obtain matching lower and upper bounds assuming $w = o(n)$ and $\eps \leq 1/8$.

\subsection{Approximate sampling: Additive model}\label{sec:additive}
Now we consider the case of additive error. Given a parameter $\veps$, we would like to sample element $i$ with probability lying in the
range $[p_i - \veps, p_i + \veps]$.
We first discuss the upper bound by giving our sampling algorithm and then give matching lower bounds. Again, assume wlog that $\frac{1}{\veps}$ is an integer. 
Let $S$ denote $\sum_j x_j$.

\vspace{0.1in}

\noindent
{\bf Upper bound}: We maintain a sorted array $A$ of size $\frac{1}{\veps}$ which stores copies of 
 numbers from 1 to $n$. For each $i$,
it stores either $\lfloor \frac{1}{\veps} \cdot \frac{x_i}{S}  \rfloor$ or $\left \lfloor \frac{1}{\veps} \cdot \frac{x_i}{S} \right \rfloor + 1$ copies of $i$. Note that this can be done because
$$ \sum_i  \left \lfloor \frac{1}{\veps} \cdot \frac{x_i}{S} \right \rfloor \leq \sum_i \frac{1}{\veps} \cdot \frac{x_i}{S} = \frac{1}{\veps}
\leq \sum_i  \left(\left \lfloor \frac{1}{\veps} \cdot \frac{x_i}{S} \right \rfloor + 1 \right) . $$
To generate a random element, the algorithm picks a uniformly random location in $A$ and outputs the number stored in that location in $A$. 
Clearly, the probability of sampling $i$ lies in the range 
$$\left[ \veps \left \lfloor \frac{1}{\veps} \cdot \frac{x_i}{S} \right \rfloor,\veps \left(\left \lfloor \frac{1}{\veps} \cdot \frac{x_i}{S} \right \rfloor+1 \right)\right] \subseteq \left[ \veps \left( \frac{1}{\veps} \cdot \frac{x_i}{S}  -1 \right),  \veps \left( \frac{1}{\veps} \cdot \frac{x_i}{S}  +1 \right) \right] = \left[p_i - \veps, p_i + \veps\right],$$
which is what we need. 
Clearly, the total space needed is the space to store $A$, i.e., $O\left(\frac{1}{\eps} \log{n}\right)$. 

\vspace{0.1in}

\noindent
{\bf Lower bound} We now prove the lower bound result. We come up with a set of distributions such that for each pair of them, 
they differ by more than $\veps$ on at least one coordinate. Consider the following set of $n$-tuples: $(x_1 \cdot \veps, x_2 \cdot \veps, \ldots, x_n \cdot \veps)$ where $x_1, ..., x_n$ are non-negative integers such that $\sum_i x_i = \frac{1}{\veps}$. If we pick any two such
distinct vectors, they will differ on at least one coordinate by at least $\veps$. 
Clearly, the size of the set of such possible vectors (or distributions)
is at least $\binom{\frac{1}{\veps} + n - 1}{n}$. 
Therefore, the space needed for sampling with $\veps$ additive error is at least
$\log \binom{\frac{1}{\veps} + n - 1}{n} \geq \Omega \left( \frac{1}{\veps} \log{n}\right)$, 
provided $\veps$ is some constant independent of $n$. 
Lower bound for smaller $\veps$ is discussed in the Appendix.
Matching these lower bounds for small $\veps$ is left as an open problem.

\section*{Acknowledgement}
RJ and AK would like to thank Karl Bringmann for discussions on Succinct Sampling.

\bibliographystyle{alpha}
\bibliography{biblio}

\appendix

\section{Proof of Lemma~\ref{lem:offline}}

\begin{proof}
Let $r$ be the smallest integer such that $2^r \geq n/\veps$, and let $k$ denote $\left \lfloor \frac{2^r}{n} \right \rfloor$.
Now we consider a sequence $x$ of $r$ random bits, and interpret this as a number between 0 and $2^r-1$. 
If this number is at least $nk$, we output $\bot$. 
Otherwise $x$ is less than $nk$. 
Let $i$ be the (unique) integer between $1$ and $n$ such that $x \in [(i-1)k, ik)$. 
In this case, the algorithm outputs the object $O_i$. 
Clearly, the probability that the algorithm outputs $O_i$ is $\frac{k}{2^r}$, which is the same for all the $n$ objects.
The probability that it outputs $\bot$ is 
$$\frac{2^r - nk}{2^r} \leq \frac{2^r - n \left( \frac{2^r}{n}-1\right)}{2^r} = \frac{n}{2^r} \leq \veps.$$
Since $r$ is $O(\log \frac{n}{\veps})$, we have shown the first part of the lemma.

\noindent
Now we prove the lower bound result. Let $R$ denote the minimum number of bits required.
Clearly, $2^R \geq n$, because there are at least $n$ possible outcomes. 
Assuming there is one sequence of random bits for which the algorithm outputs $\bot$ (recall that for a general $n$, this will be the case), we get $\veps \geq \frac{1}{2^R}$,
which implies $2^R \geq \frac{1}{\veps}$. 
Thus, $R \geq \frac{1}{2} \log \frac{n}{\veps}$.
\qed
\end{proof}

\section{Pseudocode for our Sampling Algorithms in the Streaming Setting}

Given below are the pseudocodes for the \emph{doubling} and \emph{chopping} algorithms. 
For each time $t$, we first call ${\tt Double}(t)$ and then call ${\tt Chop}(t)$. 
We maintain the variables $h,h_\bot,s,l$ and $rank$ as global variables. 
After calling ${\tt Double}(t)$ and ${\tt Chop}(t)$ for any $t$, the following properties will be satisfied:
\begin{enumerate}
\item $h=|H_1^t|=|H_2^t|=\dots =|H_t^t|$,

\item $h_\bot=|H_\bot^t|$,

\item $s=2^{r_t}$,

\item  $l$ is equal to the location of the current random string (also the index of the current stored item),

\item$rank$ is equal to the position of the current random string.
\end{enumerate}
 
 \noindent
We assume that the function ${\tt random}(y)$ returns a random integer between $1$ and $y$ and the function ${\tt random\_bit}()$ returns a random bit. 
{\tt random}$(y)$ may be easily implemented using {\tt random\_bit}$()$ when $y$ is a power of $2$ which is indeed the case below.

\begin{figure}[ht]
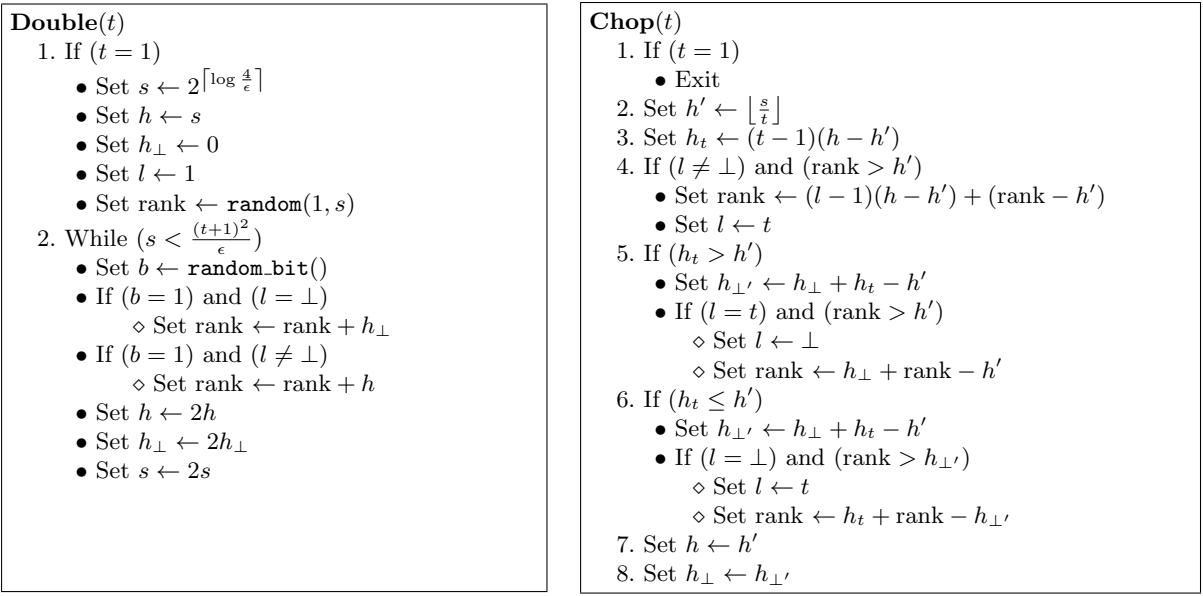

\begin{boxedminipage}{0.44\linewidth}

{\bf Double$(t)$}

\hspace{0.1in} $1.$ If $(t=1)$

\hspace{0.3in} $\bullet$ Set $s \leftarrow 2^{\left \lceil \log \frac{4}{\epsilon} \right \rceil}$

\hspace{0.3in} $\bullet$ Set $h \leftarrow s$

\hspace{0.3in} $\bullet$ Set $h_{\bot} \leftarrow 0$

\hspace{0.3in} $\bullet$ Set $l \leftarrow 1$

\hspace{0.3in} $\bullet$ Set rank $\leftarrow$ {\tt random}$(1,s)$

\hspace{0.1in} $2.$ While $(s<\frac{(t+1)^2}{\epsilon})$

\hspace{0.3in} $\bullet$ Set $b \leftarrow$  {\tt random$\_$bit}()

\hspace{0.3in} $\bullet$ If $(b=1)$ and $(l=\bot)$

\hspace{0.6in} $\diamond$ Set rank $\leftarrow \textrm{rank} + h_{\bot}$

\hspace{0.3in} $\bullet$ If $(b=1)$ and $(l\neq \bot)$

\hspace{0.6in} $\diamond$ Set rank $\leftarrow  \textrm{rank} + h$

\hspace{0.3in} $\bullet$ Set $h \leftarrow 2 h$

\hspace{0.3in} $\bullet$ Set $h_{\bot}\leftarrow 2 h_{\bot}$

\hspace{0.3in} $\bullet$ Set $s \leftarrow 2 s$

\vspace{0.55in}
\end{boxedminipage}
\quad 
\begin{boxedminipage}{0.5\linewidth}
{\bf Chop$(t)$}

\hspace{0.1in} $1.$ If $(t=1)$

\hspace{0.3in} $\bullet$ Exit

\hspace{0.1in} $2.$ Set $h^{\prime} \leftarrow \left \lfloor \frac{s}{t} \right \rfloor$

\hspace{0.1in} $3.$ Set $h_t \leftarrow (t-1)(h-h^{\prime})$

\hspace{0.1in} $4.$ If $(l \neq \bot)$ and $(\textrm{rank}>h^{\prime})$

\hspace{0.3in} $\bullet$ Set rank $\leftarrow (l-1)(h-h^{\prime})+(\textrm{rank}-h^{\prime})$

\hspace{0.3in} $\bullet$ Set $l \leftarrow t$

\hspace{0.1in} $5.$ If $(h_t > h^{\prime})$

\hspace{0.3in} $\bullet$ Set $h_{{\bot}^{\prime}} \leftarrow h_{\bot}+h_t-h^{\prime}$

\hspace{0.3in} $\bullet$ If $(l=t)$ and $(\textrm{rank} >h^{\prime})$

\hspace{0.5in} $\diamond$ Set $l \leftarrow \bot$

\hspace{0.5in} $\diamond$ Set rank $\leftarrow h_{\bot}+\textrm{rank}-h^{\prime}$

\hspace{0.1in} $6.$ If $(h_t \leq h^{\prime})$

\hspace{0.3in} $\bullet$ Set $h_{{\bot}^{\prime}} \leftarrow h_{\bot}+h_t-h^{\prime}$

\hspace{0.3in} $\bullet$ If $(l=\bot)$ and $(\textrm{rank} >h_{{\bot}^{\prime}})$

\hspace{0.5in} $\diamond$ Set $l \leftarrow t$

\hspace{0.5in} $\diamond$ Set rank $\leftarrow h_t+\textrm{rank} -h_{{\bot}^{\prime}}$

\hspace{0.1in} $7.$ Set $h \leftarrow h^{\prime}$

\hspace{0.1in} $8.$ Set $h_{\bot} \leftarrow h_{{\bot}^{\prime}}$
\end{boxedminipage}
\caption{The doubling and chopping methods}
\label{fig:chop}
\end{figure}

\noindent
Users interested in the implementation of our algorithms may find a python code at the following link useful:
\href{http://www.cse.iitd.ac.in/~rjaiswal/Research/Sampling/sampling.py}{http://www.cse.iitd.ac.in/~rjaiswal/Research/Sampling/sampling.py}.

\section{Proof of Lemma~\ref{lemma:1}}
\newcommand{\Ball}{{\tt Ball}}

\begin{proof}
Let $U \subseteq [\bits^w]^n$ denote a universe of $n$ tuples of $w$-bit numbers such that for any $(u_1, ..., u_n) \in U, \sum_i u_i = T$, where $T$ will be specified later. %% \in [n \cdot 2^w]$ the value of which will be specified later.
Given a tuple $\bar{x} = (x_1, \ldots , x_n) \in U$, $\Ball(\bar{x})$ denotes the set of  all tuples $\bar{y} \in U$ such that $\bar{y}$ is $\veps$-close to $\bar{x}$. Recall that this implies that for all $i = 1, \ldots, n$,
$$ (1-\veps) x_i \leq y_i \leq (1+\veps) x_i. $$
%Consider a tuple $\bar{y} = (y_1, ..., y_n) \in Ball(\bar{x})$. We have
%\[
%\left \lfloor \frac{x_i}{1+\eps} \right \rfloor \leq  y_i \leq \left \lceil \frac{x_i}{1-\eps} \right \rceil
%\]
So, $y_i$ can  have at most $2 \veps x_i$ different values around $x_i$. This gives the following:
\begin{equation}\label{eqn:4}
|\Ball(\bar{x})| \leq (2\veps)^n \cdot x_1 \cdot x_2 \cdots x_n  \leq 2^{wn} \cdot (2 \veps)^n,
\end{equation}
where the last inequality follows from the fact that each $x_i$ is a $w$-bit number.
%%The last inequality uses the inequality $x_1\cdot ... \cdot x_n \leq \left(\frac{x_1 + ... + x_n}{n} \right)^n$.
For any tuple $\bar{t} \in [\bits^w]^n$, the sum of elements in the tuple belongs to the set $\{0, 1, \ldots , n\cdot (2^w -1)\}$. This means that there is one value $T' \in \{0, 1, \ldots , n\cdot (2^w -1)\}$ such that the number of tuples whose sum is equal to $T'$ is at least $\frac{2^{nw}}{n \cdot (2^w-1) +1} \geq \frac{2^{nw}}{n \cdot 2^w}$. We will use $T = T'$. This implies that $|U| \geq \frac{2^{nw}}{n \cdot 2^w}$. Combining this fact with
 inequality~(\ref{eqn:4}), we get
\[
2^\mathcal{S} \geq \frac{|U|}{2^{wn} \cdot (2 \veps)^n }
\geq \frac{1}{n \cdot 2^w \cdot (2\veps)^n}
\]
This gives $\mathcal{S} \geq n \log \frac{1}{\veps} - w - n - \log n$.
\qed
\end{proof}

\section{Proof of Lemma~\ref{lemma:2}}

\begin{proof}
Let $U \subseteq ([\bits^w])^n$ denote the subset of $n$-tuples of $w$-bit numbers $\bx$ of the following form: it should be possible to
divide the $n$ coordinates in $\bx$ into $w$ blocks, each block consisting of  $n/w$ coordinates (note that these coordinates
need not be consecutive).
Let $B_l$ denote the set of indices corresponding to block $l$.
For any index $i \in B_l$, the first $(l-1)$ bits are 0, and the $l^{th}$ bit is 1 (the remaining bits can be
arbitrary).
Consider any tuple $\bx = (x_1, \ldots , x_n) \in U$. Define $\Ball(\bx)$ as the set of all tuples $\by \in U$ which
are $\veps$-close to $\bx$, i.e., for all $i = 1, \ldots, n$,
\begin{equation}\label{eqn:5}
(1 - \veps) \cdot \frac{x_i}{S} \leq \frac{y_i}{S'} \leq (1 + \veps) \cdot \frac{x_i}{S},
\end{equation}
where $S = \sum_j x_j$ and $S' = \sum_j y_j$.
Let $S_{\min} = \min_{\bar{y} \in \Ball(\bar{x})} \sum_i y_i$ and $S_{\max} = \max_{\bar{y} \in \Ball(\bar{x})} \sum_i y_i$.
Then for all $i$ we have
\[ \frac{S_{\min}}{S} \cdot (1-\veps) \cdot x_i \leq y_i \leq \frac{S_{\max}}{S} \cdot (1+\veps)\cdot x_i.
\]
Therefore, number of possible values of $y_i$ is upper bounded by
\[
\frac{x_i}{S} \cdot \left( (S_{\max} - S_{\min}) + \veps (S_{\max} + S_{\min}) \right) \leq (1+2\veps) \cdot \frac{x_i \cdot S_{\max}}{S}.
\]
Using this, we get that
\begin{equation}\label{eqn:6}
|\Ball(\bar{x})| \leq \left(\frac{S_{\max}}{S} \right)^n \cdot (x_1 \cdots  x_n) \cdot \left( 1 + 2 \veps \right)^n
\end{equation}
We will now try to get an upper bound on $|\Ball(\bar{x})|$ by obtaining suitable bounds on the quantities on the RHS of the above inequality.
First, note that due to the nature of the tuples under consideration, we have:
\[
S \geq \sum_{i=1}^{w} \frac{n}{w} \cdot 2^{w - i} = \frac{n}{w} \cdot (1 + 2 + ... + 2^{w-1}) = \frac{n}{w} \cdot (2^w - 1).
\]
Furthermore, for any $\bar{y} \in U$, we have
\[
S_{\max} \leq \sum_{i=1}^{w} \frac{n}{w} (2^{w-i+1} -1) = \frac{n}{w} (2^{w+1} - 2 - w) \leq 2 \cdot \frac{n}{w} \cdot (2^{w}-1)
\]
Next we upper bound the product of $x_1, \ldots, x_n$.
Since, each number in $i^{th}$ group is $< 2^{w - i+1}$, we can write,
\[
x_1 \cdots  x_n < \prod_{i=1}^{w} (2^{w-i+1})^{n/w} = 2^{n(w+1)/2}
\]
Putting these bounds in inequality (\ref{eqn:6}), we get that
\[
|\Ball(\bar{x})| \leq 2^n \cdot 2^{n(w+1)/2} \cdot \left( 1 + 2 \veps \right)^n
\]
Now, we try to get an estimate on $|U|$. The number of ways $w$ blocks can be arranged is $\frac{n!}{(\frac{n}{w}!)^{w}}.$
We now use the following Stirling's approximation of $a!$ for any positive integer $a$:
$$\sqrt{2\pi}  a^{a+1/2} e^{-a} \leq a! \leq e a^{a+1/2} e^{-a},$$
to get $$ \frac{n!}{(\frac{n}{w}!)^{w}} 
 \geq \sqrt{2\pi} \left( \frac{w}{e^2 n}\right)^{\frac{w}{2}} n^{1/2} w^n.$$ 
So, we get
\begin{eqnarray*}
|U| \geq \left( \sqrt{2\pi} \left( \frac{w}{e^2 n}\right)^{\frac{w}{2}} n^{1/2} w^n \right) \cdot \prod_{i=1}^{w} (2^{w - i})^{\frac{n}{w}} 
=  \left( \sqrt{2\pi} \left( \frac{w}{e^2 n}\right)^{\frac{w}{2}} n^{1/2} w^n \right) \cdot 2^{n(w-1)/2}
\end{eqnarray*}
Using this bound, we have
\begin{eqnarray*}
2^{\mathcal{S}} &\geq& \frac{|U|}{2^n \cdot 2^{n(w+1)/2} \cdot \left( 1 + 2 \veps \right)^n} \\
&\geq& \frac{\left( \sqrt{2\pi} \left( \frac{w}{e^2 n}\right)^{\frac{w}{2}} n^{1/2} w^n \right) \cdot 2^{n(w-1)/2}}{2^n \cdot 2^{n(w+1)/2} \cdot \left(1 + 2 \eps\right)^n} \\
&=& \frac{1}{2^{2n}} \cdot \sqrt{2 \pi n} \cdot w^n \cdot \left(\frac{w}{e^2 n}\right)^{w/2} \cdot \left(\frac{1}{1+2\veps} \right)^n 
\end{eqnarray*}
which implies that 
\begin{eqnarray*}
{\cal S} &\geq& n \log{w} + \log{\sqrt{2 \pi n}} + \frac{w}{2} \log{\frac{w}{e^2 n}} + n \log{\frac{1}{4 (1 + 2 \veps)}} \\
&\geq& n \log w - n \log{4(1+2 \veps)} -\frac{w}{2} \log{(e^2 n)} 
\end{eqnarray*}
%%\Rightarrow \mathcal{S} &\geq& n \log{w} + \log{\sqrt{2 \pi n}} + \frac{w}{2} \log{\frac{w}{e^2 n}} + n \log{\frac{1}{4 {1 + 2 \eps}} \\
%%\Rightarrow \mathcal{S} &\geq& n \log{w} - n \log{\frac{4 (1 + 2 \eps)}{1-\eps^2}} - \frac{w}{2} \log{(e^2 n)}.
%%\end{eqnarray*}
This concludes the proof of the lemma.
\qed
\end{proof}

\section{Additive model: Lower bounds}
As shown in Section~\ref{sec:additive}, the lower bound on space is given by the expression:
\begin{eqnarray*}
\mathcal{S} \geq \log{\binom{1/\veps + n - 1}{n}} = \Omega \left( \frac{1}{\veps} \cdot \log{(1 + \veps n)} + n \log{\left(1 + \frac{1}{\veps n} \right)}\right).
\end{eqnarray*}
So, we get the following lower bounds in the following two cases:
\begin{enumerate}
\item $\veps \geq 1/n$: In this case, we get that $\mathcal{S} = \Omega \left(\frac{1}{\veps} \cdot \log{\veps n}\right)$.

\item $\veps < 1/n$: In this case, we get that $\mathcal{S} = \Omega \left(n \cdot \log{\frac{1}{\veps n}}\right)$.
\end{enumerate}

%\noindent
%Using simple optimisations to our techniques in Section~\ref{sec:additive}, we can an upper bound of $O(\frac{1}{\veps} + n)$.
%Note that this is better than our bound of $O(\frac{1}{\veps} \cdot \log{n})$ given in Section~\ref{sec:additive} when $\veps = \Omega \left(\frac{\log{n}}{n} \right)$.
%This new upper bound follows from the simple observation that

\end{document}